\documentclass[11pt]{article}
\usepackage{fullpage}
\usepackage{times}
\usepackage{subfigure}
\usepackage{color}
\usepackage{url}
\usepackage{graphicx}
\usepackage{amsmath}
\usepackage{amssymb}

\newcommand{\qed}{\mbox{}\hspace*{\fill}\nolinebreak\mbox{$\rule{0.6em}{0.6em}$}
}
\newcommand{\expect}{{\bf \mbox{\bf E}}}
\newcommand{\prob}{{\bf \mbox{\bf Pr}}}

\definecolor{gray}{rgb}{0.5,0.5,0.5}
\newcommand{\e}{{\epsilon}}
\newcommand{\eps}{{\epsilon}}

\newtheorem{theorem}{Theorem}[section]
\newtheorem{lemma}[theorem]{Lemma}

\newtheorem{corollary}[theorem]{Corollary}
\newtheorem{definition}[theorem]{Definition}
\newtheorem{remark}[theorem]{Remark}

\newenvironment{proof}{{\bf Proof:}}{$\qed$\par}

\newenvironment{proofof}[1]{\noindent{\bf Proof of #1:}}{$\qed$\par}

\usepackage{hyperref}
\hypersetup{
bookmarksnumbered
}

\newcommand{\refine}{{\bf \mbox{\sc{Refine}}}}

\newcommand{\find}{{\bf \mbox{\sc{Find}}}}
\newcommand{\union}{{\bf \mbox{\sc{Union}}}}
\begin{document}

\title{\Large Graph Sparsification via Refinement Sampling}
\setcounter{page}{0}
\author{Ashish Goel\thanks{
    Departments of Management Science and Engineering and (by courtesy)
    Computer Science, Stanford University.
    Email: {\tt ashishg@stanford.edu}.
    Research supported in part by NSF award IIS-0904325.}\\
\and Michael Kapralov\thanks{
    Institute for Computational and Mathematical Engineering, Stanford University.
    Email: {\tt kapralov@stanford.edu}. Research supported by a Stanford Graduate Fellowship.}\\
\and Sanjeev Khanna\thanks{Department of Computer and Information Science, University of Pennsylvania,
Philadelphia PA. Email: {\tt sanjeev@cis.upenn.edu}.
Supported in
part by NSF Awards CCF-0635084 and IIS-0904314.}
}

\maketitle
\thispagestyle{empty}
\pdfbookmark[1]{Abstract}{MyAbstract}
\begin{abstract}
  A graph $G'(V,E')$ is an {\em $\eps$-sparsification} of $G$ for some $\eps >
  0$, if every (weighted) cut in $G'$ is within $(1\pm \e)$ of the
  corresponding cut in $G$. A celebrated result of Bencz\'{u}r and Karger
  shows that for every undirected graph $G$, an $\eps$-sparsification with
  $O(n \log n/\e^2)$ edges can be constructed in $O(m\log^2n)$ time. The
  notion of cut-preserving graph sparsification has played an important role
  in speeding up algorithms for several fundamental network design and routing
  problems. Applications to modern massive data sets often constrain
  algorithms to use computation models that restrict random access to the
  input. The semi-streaming model, in which the algorithm is constrained to
  use $\tilde O(n)$ space, has been shown to be a good abstraction for
  analyzing graph algorithms in applications to large data sets. Recently, a
  semi-streaming algorithm for graph sparsification was presented by Anh and
  Guha; the total running time of their implementation is $\Omega(mn)$, too
  large for applications where both space and time are important. In this
  paper, we introduce a new technique for graph sparsification, namely {\em
    refinement sampling}, that gives an $\tilde{O}(m)$ time semi-streaming
  algorithm for graph sparsification.

  Specifically, we show that refinement sampling can be used to design a
  one-pass streaming algorithm for sparsification that takes $O(\log\log n)$
  time per edge, uses $O(\log^2 n)$ space per node, and outputs an
  $\eps$-sparsifier with $O(n\log^3 n/\e^2)$ edges.  At a slightly increased
  space and time complexity, we can reduce the sparsifier size to $O(n \log
  n/\e^2)$ edges matching the Bencz\'{u}r-Karger result, while improving upon
  the Bencz\'{u}r-Karger runtime for $m=\omega(n\log^3 n)$.  Finally, we show
  that an $\eps$-sparsifier with $O(n \log n/\e^2)$ edges can be constructed
  in two passes over the data and $O(m)$ time whenever $m =
  \Omega(n^{1+\delta})$ for some constant $\delta > 0$. As a by-product of our approach,
  we also obtain an $O(m \log\log n+ n \log n)$ time streaming algorithm to compute
  a sparse $k$-connectivity certificate of a graph.

\end{abstract}

\newpage
\section{Introduction}

The notion of graph sparsification was introduced in \cite{benczurkarger96},
where the authors gave a near linear time procedure that takes as input an
undirected graph $G$ on $n$ vertices and constructs a weighted subgraph $H$ of
$G$ with $O(n\log n/\e^2)$ edges such that the value of every cut in $H$ is
within a $1\pm \e$ factor of the value of the corresponding cut in $G$. This
algorithm has subsequently been used to speed up algorithms for finding
approximately minimum or sparsest cuts in graphs (\cite{benczurkarger96,
  krv06}), as well as in a host of other applications (e.g. \cite{kl02}). A
more general class of spectral sparsifiers was recently introduced by Spielman
and Srivastava in \cite{ss:sample2008}. The algorithms developed in
\cite{benczurkarger96} and \cite{ss:sample2008} take near-linear time in the
size of the graph and produce very high quality sparsifiers, but require
random access to the edges of the input graph $G$, which is often
prohibitively expensive in applications to modern massive data sets. The
streaming model of computation, which restricts algorithms to use a small
number of passes over the input and space polylogarithmic in the size of the
input, has been studied extensively in various application domains
(e.g. \cite{b:streaming}), but has proven too restrictive for even the
simplest graph algorithms (even testing $s-t$ connectivity requires
$\Omega(n)$ space). The less restrictive semi-streaming model, in which the
algorithm is restricted to use $\tilde O(n)$ space, is more suited for graph
algorithms~\cite{fkmsz05}. The problem of constructing graph sparsifiers in
the semi-streaming model was recently posed by Anh and Guha~\cite{anh-guha},
who gave a one-pass algorithm for finding Bencz\'{u}r-Karger type sparsifiers
with a slightly larger number of edges than the original Bencz\'{u}r-Karger
algorithm, i.e. $O(n\log n\log\frac{m}{n}/\e^2)$ as opposed to $O(n\log
n/\e^2)$. Their algorithm requires only one pass over the data, and their
analysis is quite non-trivial. However, its time complexity is
$\Omega(mn\mbox{\ polylog}(n))$, making it impractical for applications where
both time and space are important constraints\footnote{As is often the case
  for semi-streaming algorithms, Anh and Guha do not explicitly compute the
  running time of their algorithm; $\Omega(mn\mbox{\ polylog}(n))$ is the best
  running time we can come up with for their algorithm.}

Apart from the issue of random access vs disk, the semi-streaming model is
also important for scenarios where edges of the graph are
revealed one at a time by an external process. For example, this application maps
well to online social networks where edges arrive one by one, but efficient
network computations may be required at any time, making it particularly
useful to have a dynamically maintained sparsifier.

\paragraph{Our results:} We introduce the concept of \emph{refinement
  sampling}. At a high level, the basic idea is to sample edges at
geometrically decreasing rates, using the sampled edges at each rate to refine
the connected components from the previous rate. The sampling rate at which
the two endpoints of an edge get separated into different connected components
is used as an approximate measure of the ``strength'' of that edge.  We use
refinement sampling to obtain two algorithms for computing Bencz\'{u}r-Karger
type sparsifiers of undirected graphs in the semi-streaming model
efficiently. The first algorithm requires $O(\log n)$ passes, $O(\log n)$
space per node, $O(\log n\log\log n)$ work per edge and produces sparsifiers
with $O(n\log^2 n/\e^2)$ edges. The second algorithm requires one pass over
the edges of the graph, $O(\log^2 n)$ space per node, $O(\log\log n)$ work per
edge and produces sparsifiers with $O(n\log^3 n/\e^2)$ edges. Several
properties of these results are worth noting:
\begin{enumerate}
\item In the incremental model, the amortized running time per edge arrival is
  $O(\log \log n)$, which is quite practical and much better than the
  previously best known running time of $\Omega(n)$.
\item The sample size can be improved for both algorithms by running the original
Benc\'{u}r-Karger algorithm on the sampled graph without violating the
restrictions of the semi-streaming model, yielding $O(\log n\log \log
n+(\frac{n}{m})\log^4 n)$ and $O(\log \log n+(\frac{n}{m})\log^5 n)$ amortized work per edge
respectively.
\item Somewhat surprisingly, this two-stage (but still semi-streaming)
  algorithm improves upon the runtime of the original sparsification scheme
  when $m=\omega(n\log^2n)$ for the $O(\log n)$-pass version and
  $m=\omega(n\log^3 n)$ for the one-pass version.
\item As a by-product of our analysis, we show that refinement sampling can be
  regarded as a one-pass algorithm for producing a sparse connectivity
  certificate of a weighted undirected graph (see Corollary
  \ref{cor:sparse-cert}). Thus we obtaining a streaming analog 
  of the Nagamochi-Ibaraki result~\cite{nagamochi-ibaraki} for producing sparse
  certificates, which is in turn used in the Benc\'{u}r-Karger sampling.
\end{enumerate}

Finally, in Section \ref{sec:two-pass} we give an algorithm for
constructing $O(n\log n/\e^2)$-size sparsifiers in $O(m)$ time using two
passes over the input when $m=\Omega(n^{1+\delta})$.

\noindent
\paragraph{Related Work:}
In \cite{anh-guha} the authors give an algorithm
for sparsification in the semi-streaming model based on the observation that one can use the constructed sparsification of the currently received part of the graph to estimate of the strong connectivity of a newly received edge. A brief outline of the algorithm is as follows. Denote the edges of $G$ in their order in the stream by $e_1,\ldots, e_m$. Set $H_0=(V, \emptyset)$. For every $t>0$ compute the strength $s_t$ of $e_t$ in $H_{t-1}$, and with probability $p_{e_t}=\min\{\rho/s_t, 1\}$ set $H_{t}=(V, E(H_{t-1})\cup \{e_t\})$, giving $e_t$ weight $1/p_{e_t}$ in $H_{t}$ and $H_t=H_{t-1}$ otherwise. For every $t$ the graph $H_t$ is an $\e$-sparsification of the subgraph received by time $t$.
The authors show that this algorithm yields an $\e$-sparsifier with $O(n\log n\log\frac{m}{n}/\e^2)$ edges. However, it is unclear how one can calculate the strengths $s_t$ efficiently. A naive implementation would take $\Omega(n)$ time for each $t$, resulting in $\Omega(mn)$ time overall. One could conceivably use the fact that $H_{t-1}$ is always a subgraph of  $H_{t}$, but to the best of our knowledge there are no results on efficiently calculating or approximating \emph{strong} connectivities in the incremental model.

 It is important to emphasize that our techniques for obtaining an efficient one-pass sparsification algorithm are very different from the approach of \cite{anh-guha}. In particular, the structure of dependencies in the sampling process is quite different. In the algorithm of \cite{anh-guha} edges are not sampled independently since the probability with which an edge is sampled depends on the the coin tosses for edges that came earlier in the stream. Our approach, on the other hand, decouples the process of estimating edge strengths from the process of producing the output sample, thus simplifying analysis and making a direct invocation of the Bencz\'{u}r-Karger sampling theorem possible.

\noindent
\paragraph{Organization:}
Section~\ref{sec:prelim} introduces some notation as well as reviews the
Bencz\'{u}r-Karger sampling algorithm. We then introduce in Section~\ref{sec:refinement}
our \emph{refinement sampling} scheme, and show how it can be used to obtain a sparsification algorithm requiring $O(\log n)$ passes and $O(\log n\log\log n)$ work per edge. The size of the sampled graph is $O(n\log^2 n/\e^2)$, i.e. at most $O(\log n)$ times larger than that produced by Bencz\'{u}r-Karger sampling.
Finally, in Section \ref{sec:onepass} we build on the ideas of Section~\ref{sec:refinement}
to obtain a one-pass algorithm with $O(\log \log n)$ work per edge at the expense of increasing the size of the sample to $O(n\log^3 n/\e^2)$.

\section{Preliminaries}
\label{sec:prelim}

We will denote by $G(V, E)$ the input undirected graph with vertex set $V$ and edge set $E$ with $|V|=n$ and $|E|=m$.
For any $\eps > 0$, we say that a weighted graph $G'(V,E')$ is an  {\em $\eps$-sparsification}
of $G$ if every (weighted) cut in $G'$ is within $(1\pm \e)$ of the corresponding
cut in $G$.
Given any two collections of sets that partition $V$, say $S_1$ and $S_2$,
we say that
$S_2$ is a {\em refinement} of $S_1$ if for any $X \in S_1$ and $Y \in S_2$,
either $X \cap Y = \emptyset$ or $Y \subset X$. In other words, $S_1 \cup S_2$
form a laminar set system.

\subsection{Bencz\'{u}r-Karger Sampling Scheme}

We say that a graph is \emph{$k$-connected} if the value of each cut in $G$ is at least $k$.
The Bencz\'{u}r-Karger sampling scheme uses a more strict notion of connectivity, referred to as
\emph{strong connectivity}, defined as follows:

\begin{definition}\cite{benczurkarger96}
A \emph{$k$-strong component} is a maximal $k$-connected vertex-induced subgraph.
The \emph{strong connectivity} of an edge $e$, denoted by $s_e$, is the largest $k$ such that a $k$-strong component contains $e$.
\end{definition}

Note that the set of $k$-strong components form a partition of the vertex set of $G$, and the set of $k+1$-strong components forms a refinement this partition.
We say $e$ is \emph{$k$-strong} if its strong connectivity is $k$ or more, and \emph{$k$-weak} otherwise.
The following simple lemma will be useful in our analysis.

\begin{lemma}\cite{benczurkarger96}\label{lm:k-weak}
The number of $k$-weak edges in a graph on $n$ vertices is bounded by $k(n-1)$.
\end{lemma}

The sampling algorithm relies on the following result:

\begin{theorem}\cite{benczurkarger96} \label{thm:bk-sampling}
Let $G'$ be obtained by sampling edges of $G$ with probability $p_e=\min\{\frac{\rho}{\e^2 s_e}, 1\}$, where $\rho=16(d+2)\ln n$, and giving each sampled edge weight $1/p_e$. Then
$G'$ is an $\eps$-sparsification of $G$ with probability at least $1-n^{-d}$.
Moreover, expected number of edges in $G'$ is $O(n \log n)$.
\end{theorem}

It follows easily from the proof of theorem \ref{thm:bk-sampling} in \cite{benczurkarger96} that
if we sample using an {\em underestimate} of edge strengths, the resulting graph is
still an $\eps$-sparsification.

\begin{corollary}\label{cor:oversampling}
Let $G'$ be obtained by sampling each edge of $G$ with probability $\tilde p_e\geq p_e$ and and give every sampled edge $e$ weight $1/\tilde p_e$.
Then
$G'$ is an $\eps$-sparsification of $G$ with probability at least $1-n^{-d}$.
\end{corollary}

In \cite{benczurkarger96} the authors give an $O(m\log^2 n)$ time algorithm for calculating estimates of strong connectivities that are sufficient for sampling. The algorithm, however, requires random access to the edges of the graph, which is disallowed in the semi-streaming model. More precisely,
the procedure for estimating edge strengths given in \cite{benczurkarger96} relies on the Nagamochi-Ibaraki algorithm for obtaining sparse certificates for edge-connectivity
in $O(m)$ time (\cite{nagamochi-ibaraki}). The algorithm of \cite{nagamochi-ibaraki} relies on random access to edges of the graph and to the best of our knowledge no streaming implementation is known. In fact we show in Corollary \ref{cor:sparse-cert} that refinement sampling yields a streaming algorithm for producing sparse certificates for edge-connectivity in one pass over the data.

In what follows we will consider unweighted graphs to simplify notation. The results obtained can  be easily extended to the polynomially weighted case as outlined in Remark \ref{rmk:weights} at the end of Section \ref{sec:onepass}.

\section{Refinement Sampling}
\label{sec:refinement}

We start by introducing the idea of refinement sampling that gives a simple algorithm
for efficiently computing a BK-sample, and serves as a building block for our streaming algorithms.

To motivate refinement sampling, let us consider the simpler problem of
identifying all edges of strength at least $k$ in the input graph $G(V,E)$. A natural idea to do so
is as follows: (a) generate a graph $G'$ by sampling edges of $G$ with probability $\tilde{O}(1/k)$,
(b) find connected components of $G'$, and (c) output all edges $(u,v) \in E$ as such that
$u$ and $v$ are in the same connected component in $G'$.
The sampling rate of $\tilde{O}(1/k)$ suggests that if an edge $(u,v)$ has
strong connectivity below $k$, the vertices $u$ and $v$ would end up in different components in $G'$,
and conversely, if the strong connectivity of $(u,v)$ is above $k$, they are likely to stay connected
and hence output in step $(c)$.
While this process indeed filters out most $k$-weak edges,
it is easy to construct examples where the output will contain many edges of strength $1$ even
though $k$ is polynomially large (a star graph, for instance).
The idea of refinement sampling is to get around this by successively {\em refining} the sample
obtained in the final step $(c)$ above.

In designing our algorithm, we will repeatedly invoke the subroutine
\refine$(S, p)$ that essentially implements the simple idea described above.

\begin{description}
\item[Function:] \refine$(S, p)$
\item[Input:] Partition $S$ of $V$, sampling probability $p$.
\item[Output:] Partition $S'$ of $V$, a refinement of $S$.
\item[1.] Take a uniform sample $E'$ of edges of $E$ with probability $p$.
\item[2.] For each $U\in S, U\subseteq V$ let $C(U)$ be the set of connected components of $U$ induced by $E'$.
\item[3.] Return $S':=\cup_{U\in S} C(U)$.
\end{description}

It is easy to see that \refine~can be implemented using $O(n)$ space, a total of $n$ \union~ operations with $O(n\log n)$ overall cost and $m$ \find~ operations with $O(1)$ cost per operation, for an overall running time of
$O(n \log n + m)$(see, e.g. \cite{b:clr}).
Also, \refine ~can be implemented using a single pass over the set of edges.
A scheme of refinement relations between $S_{l,k}$ is given in Fig. \ref{fig:f1}.

The {\bf refinement sampling} algorithm computes partitions $S_{l,j}$ for $l=1,\ldots,L$ and $j=0, 1, \ldots,K$. Here $L=\log (2n)$ is the number of strength levels (the factor of $2$ is chosen for convenience to ensure that $S_{L, K}$ consists of isolated vertices whp), $K$ is a parameter which we call the {\em strengthening} parameter. Also, we choose a parameter $\phi> 0$, which we will refer to as the oversampling parameter. For a partition $S$, let $X(S)$ denote all the edges in $E$ which have endpoints in two different sets in $S$. The partitions are computed as follows:

\begin{description}
\item[Algorithm 1 (Refinement Sampling)]
\item[Initialization:] $S_{l,0} = \{V \}$ for $l=1,\ldots,L$.
\item[1.] Set $k:=1$
\item[2.] For each $l$, $1\leq l\leq L$, set $S_{l,k}:=\refine (S_{l,k-1}, 2^{-l})$.
\item[3.] Set $k:=k+1$. If $k<K$, go to step 1.
\item[4.] For each $e\in E$ define $L(e) = \min\left\lbrace l : e \in X(S_{l,K})\right\rbrace$. Sample edge $e$ with probability $z(e) = \min\{1, \frac{\phi}{\e^2 2^{L(e)}}\}$ and assign it weight $1/z(e)$. Let $R(\phi,K)$ denote the set of edges sampled during this step; we call this the
refinement sample of $G$.
\end{description}

The following two lemmas relate the probabilities $z(e)$ to the sampling probabilities used in the Bencz\'{u}r-Karger sampling scheme.
\begin{lemma}\label{lm:upper}
For any $K>0$, with probability at least $1-K n^{-d}$ every edge $e$ satisfies $z(e) \leq 4\phi \rho/(\e^2 s_e)$.
\end{lemma}
\begin{proof}
Consider an edge $e$ with strong connectivity $s_e$, and let $C$ denote the $s_e$-strongly connected component containing $e$. By Theorem \ref{thm:bk-sampling}, sampling with probability $\min\{4\rho/s_e, 1\}$ preserves all cuts up to $1\pm \frac1{2}$ in $C$ with probability at least $1-n^{-d}$. Hence, all $s_e$-strongly connected components stay connected after $K$ passes of \refine~for all $l>0$ such that $2^{-l}\geq 4\rho/s_e$, yielding the lemma.
\end{proof}

\begin{lemma}\label{lm:lower}
If $K > \log_{4/3} n$, then $2^{-L(e)+1}\geq 1/(2s_e)$ for every $e\in E(G)$ with probability at least $1-Ke^{-(n-1)/100}$.
\end{lemma}
\begin{proof}
Consider a level $l$ such that $p=2^{-l} < 1/(2s_e)$. Let $H$ be the graph
obtained by contracting all $(s_e+1)$-strong components in $G$ into supernodes.
Since $H$ contains only  $(s_e+1)$-weak edges, the number of edges is at most $s_e(n-1)$ by Lemma \ref{lm:k-weak}. As the expected number of $(s_e+1)$-weak edges in the sample is at most $(n-1)/2$, by
Chernoff bounds, the probability that the
number of $(s_e+1)$-weak edges in the sample exceeds $3(n-1)/4$ is at most $(e^{1/4}(5/4)^{-5/4})^{-(n-1)/2}<e^{-(n-1)/100}$. Thus
at least one quarter of the supernodes get isolated in each iteration. Hence, no $(s_e+1)$-weak edge survives after $K=\log_{4/3} n$ rounds of refinement sampling with probability at least $1-K e^{-(n-1)/100}$. Since $L(e)$ was defined as the least $l$ such that $e\in X(S_{l, K})$, the endpoints of $e$ were connected in $S_{L(e)-1, K}$, so $2^{-L(e)+1}\geq 1/(2s_e)$.
\end{proof}

\begin{theorem}
Let $G'$ be the graph obtained by running Algorithm 1 with $\phi:=4\rho$. Then $G'$ has
$O(n\log^2 n/\e^2)$ edges in expectation, and is an
$\eps$-sparsification of $G$ with probability at least $1-n^{-d+1}$.
\end{theorem}
\begin{proof}
We have  from lemma \ref{lm:lower} and the choice of $\phi$ that the sampling probabilities dominate those used in Bencz\'{u}r-Karger sampling with probability at least $1-K e^{-(n-1)/100}$. Hence, by corollary \ref{cor:oversampling}
we have that every cut in $G'$ is within $1\pm \eps$ of its value in $G$ with probability at least $1-K e^{-(n-1)/100}-n^{-d}$. The expected size of the sample is $O(n\log^2 n/\e^2)$ by lemma \ref{lm:upper} together with the fact that $\rho=O(\log n)$. The probability of failure of the estimate in lemma \ref{lm:lower} is at most $Kn^{-d}$, so all bounds hold with probability at least $1-Kn^{-d}+K e^{-(n-1)/100}-n^{-d}>1-n^{-d+1}$ for sufficiently large $n$.  The high probability bound on the number of edges follows by an application of the Chernoff bound.
\end{proof}

\begin{figure}
\begin{picture}(20,10)(0,0)
\put(1, 8){\framebox(1.7, 1){$S_{1,1}$}}
\put(4, 8){\framebox(1.7, 1){$S_{1,2}$}}
\put(7, 8){\makebox(1.7, 1){$\ldots$}}
\put(10, 8){\framebox(1.7, 1){$S_{1,K-1}$}}
\put(13, 8){\framebox(1.7, 1){$S_{1,K}$}}

\put(2.7,8.5){\vector(1, 0){1.3}}
\put(5.7,8.5){\vector(1, 0){1.3}}
\put(8.7,8.5){\vector(1, 0){1.3}}
\put(11.7,8.5){\vector(1, 0){1.3}}

\put(1, 6){\framebox(1.7, 1){$S_{2,1}$}}
\put(4, 6){\framebox(1.7, 1){$S_{2,2}$}}
\put(7, 6){\makebox(1.7, 1){$\ldots$}}
\put(10, 6){\framebox(1.7, 1){$S_{2,K-1}$}}
\put(13, 6){\framebox(1.7, 1){$S_{2,K}$}}

\put(2.7,6.5){\vector(1, 0){1.3}}
\put(5.7,6.5){\vector(1, 0){1.3}}
\put(8.7,6.5){\vector(1, 0){1.3}}
\put(11.7,6.5){\vector(1, 0){1.3}}

\put(1, 5){\makebox(1.7, 1){$\vdots$}}
\put(4, 5){\makebox(1.7, 1){$\vdots$}}
\put(7, 5){\makebox(1.7, 1){$\vdots$}}
\put(10, 5){\makebox(1.7, 1){$\vdots$}}
\put(13, 5){\makebox(1.7, 1){$\vdots$}}

\put(1, 4){\framebox(1.7, 1){$S_{L-1,1}$}}
\put(4, 4){\framebox(1.7, 1){$S_{L-1,2}$}}
\put(7, 4){\makebox(1.7, 1){$\ldots$}}
\put(10, 4){\framebox(1.7, 1){$S_{L-1,K-1}$}}
\put(13, 4){\framebox(1.7, 1){$S_{L-1,K}$}}

\put(2.7,4.5){\vector(1, 0){1.3}}
\put(5.7,4.5){\vector(1, 0){1.3}}
\put(8.7,4.5){\vector(1, 0){1.3}}
\put(11.7,4.5){\vector(1, 0){1.3}}

\put(1, 2){\framebox(1.7, 1){$S_{L,1}$}}
\put(4, 2){\framebox(1.7, 1){$S_{L,2}$}}
\put(7, 2){\makebox(1.7, 1){$\ldots$}}
\put(10, 2){\framebox(1.7, 1){$S_{L,K-1}$}}
\put(13, 2){\framebox(1.7, 1){$S_{L,K}$}}

\put(2.7,2.5){\vector(1, 0){1.3}}
\put(5.7,2.5){\vector(1, 0){1.3}}
\put(8.7,2.5){\vector(1, 0){1.3}}
\put(11.7,2.5){\vector(1, 0){1.3}}
\end{picture}
\caption{Scheme of refinement relations between partitions for Algorithm 1.}
\label{fig:f1}
\end{figure}
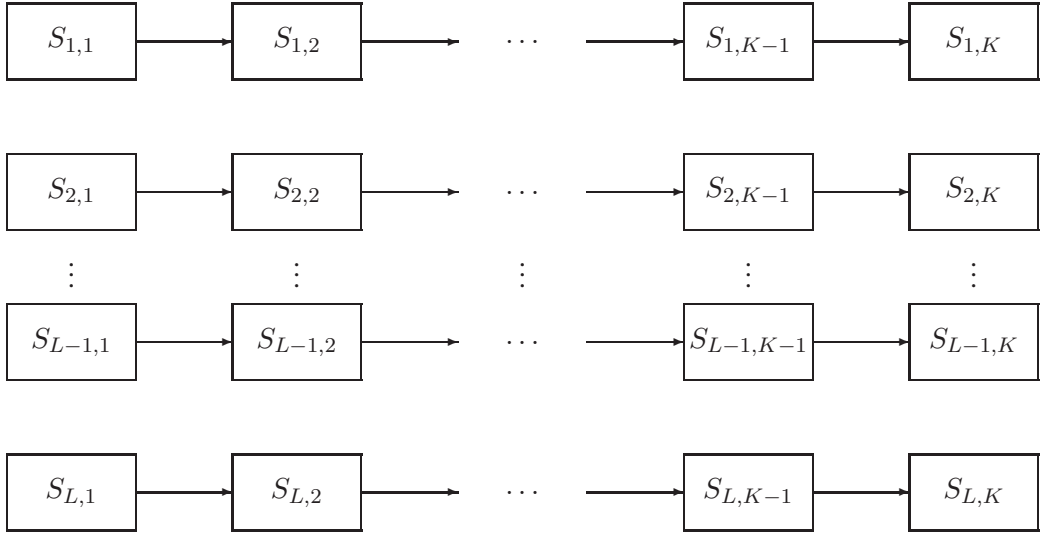

\smallskip
The next lemma follows from the discussion above:

\begin{lemma}
For any $\eps > 0$, an $\eps$-sparsification of $G$ with $O(n \log^2 n/\eps^2)$ edges
can be constructed in $O(\log n)$ passes of \refine\ using
$O(\log n)$ space per node and $O(\log^2 n)$ time per edge.
\end{lemma}

We now note that one $\log n$ factor in the running time comes from the fact that during each pass $k$ Algorithm 1 flips a coin at every level $l$ to decide whether or not to include $e$ into $S_{l, k}$ when $e\in S_{l, k-1}$. If we could guarantee that $S_{l, k}$ is a refinement of $S_{l', k}$ for all $l'<l$ and for all $k$, we would be able to use binary search to find the largest $l$ such that $e\in S_{l, k}$ in $O(\log \log n)$ time.
Algorithm 2 given below uses iterative sampling to ensure a scheme of refinement relations given in Fig. \ref{fig:f2}. For each edge $e$, $1 \le k \le K$, and $1 \le \ell \le L$, we
define for convenience independent Bernoulli random variables $A_{l,k,e}$ such
$\prob[A_{l,k,e}=1]=1/2$, even though the algorithm will not always need to flip all these $O(\log^2 n)$ coins. Also define $U_{l,k,e}=\prod_{j\leq l} A_{j,k,e}$.
The algorithm uses connectivity data structures $D_{l,k}$, $1\leq l\leq L, 1\leq k\leq K$.
Adding an edge $e$ to $D_{l,k}$ merges the components that the endpoints of $e$ belong to in $D_{l,k}$.

\begin{description}
\item[Algorithm 2 (An $O(\log n)$-Pass Sparsifier)]
\item[Input:] Edges of $G$ streamed in adversarial order: $(e_1,\ldots, e_m)$.
\item[Output:] A sparsification $G'$ of $G$.
\item[Initialization:]  Set $E':=\emptyset$.
\item[1.] For all $k=1,\ldots, K$
\item[2.] Set $t=1$.
\item[3.] For all $l=1,\ldots, L$
\item[4.] Add $e_t=(u_t, v_t)$ to $D_{l, k}$ if $U_{l,k,e}=1$ and $u_t$ and $v_t$ are connected in $D_{(l, k-1)}$.
\item[5.] Set $t:=t+1$. Go to step 1 if $t\leq m$.
\item[6.] For each $e_t$ define $L'(e_t)$ as the minimum $l$ such that $u_t$ and $v_t$ are not connected in $D_{l, K}$. Set $z'(e_t):=\min\left\{1, \frac{4\rho}{\e^2 2^{L'(e_t)}}\right\}$. Output $e_t$ with probability $z'(e_t)$, giving it weight $1/z'(e_t)$.
\end{description}

\begin{theorem}
For any $\eps > 0$, there exists an $O(\log n)$-pass streaming algorithm that produces an $\eps$-sparsification $G'$ of a graph $G$ with at most $O(n\log^2 n/\e^2)$ edges using
$O((n/m) \log n + \log n \log \log n)$ time per edge.
\end{theorem}
\begin{proof}
The correctness of Algorithm 2 follows in the same way as for Algorithm 1 above, so it remains to determine its runtime. An $O((n/m)\log n+1)$ term per edge comes from amortized $O(n\log n+m)$ complexity of UNION-FIND operations. The $\log n$ factor in the runtime comes from the $\log n$ passes, and we now show that step 3 can be implemented in $O(\log \log n)$ time. First note that since  $S_{l',k'}$ is a refinement of $S_{l,k}$ whenever $l'\geq l$ and $k'\geq k$,  one can use binary search to determine the largest $l_0$ such that $u_t$ and $v_t$ are connected in $D_{l_0-1,k-1}$. One then keeps flipping a fair coin
and adding $e$ to connectivity data structures $D_{l, k}$ for successive $l\geq l_0$ as long as the coin keeps coming up heads. Since $2$ such steps are performed on average, it takes $O(K)=O(\log n)$ amortized time per edge by the Chernoff bound. Putting these estimates together, we obtain the claimed time complexity.
\end{proof}

The scheme of refinement relations between $S_{l,k}$ is depicted in Fig. \ref{fig:f2}.

\begin{corollary}
For any $\eps > 0$, there is an $O(\log n)$-pass algorithm
that produces an $\eps$-sparsification $G'$ of an input graph $G$ with at
most $O(n\log n/\e^2)$ edges using $O(\log^2 n)$ space per node,
and performing $O(\log n\log\log n+(n/m)\log^4 n)$ amortized work per edge.
\end{corollary}
\begin{proof}
One can obtain a sparsification $G'$ with $O(n \log^2 n/\e^2)$ edges by running Algorithm 2 on the input graph $G$, and then run the Bencz\'{u}r-Karger algorithm
on $G'$ without violating the restrictions of the semi-streaming model. Note that even though $G'$ is a weighted graph, this will have overhead $O(\log^2 n)$ per edge of $G'$ since the weights are polynomial. Since $G'$ has $O(n\log^2 n)$ edges, the amortized work per edge of $G$ is $O(\log n\log\log n+(n/m)\log^4 n)$.
The Bencz\'{u}r-Karger algorithm can be implemented using space proportional to the size of the graph, which yields $O(\log^2 n)$ space per node.
\end{proof}
\begin{remark}
The algorithm improves upon the runtime of  the Bencz\'{u}r-Karger sparsification scheme when $m=\omega(n\log^2 n)$.
\end{remark}

\section{A One-pass $\tilde{O}(n+m)$-Time Algorithm for Graph Sparsification} \label{sec:onepass}

In this section we convert Algorithm 2 obtained in the previous section to a one-pass algorithm.
We will design a one-pass algorithm that produces an
$\eps$-sparsifier with $O(n\log^3 n/\eps^2)$ edges using only
$O(\log \log n)$ amortized work per edge. A simple post-processing step at the
end of the algorithm will allow us to reduce the size to $O(n\log n/\eps^2)$ edges
with a slightly increased space and time complexity. The main difficulty is that in going
from $O(\log n)$ passes to a one-pass algorithm,
we need to introduce and analyze new dependencies in the sampling process.

As before, the algorithm maintains connectivity data structures $D_{l, k}$, where $1\leq l\leq L$ and $1\leq k\leq K$. In addition to indexing $D_{l,k}$ by pairs $(l,k)$ we shall also write $D_{J}$ for $D_{l,k}$, where $J=K(l-1)+k$, so that $1\leq J \leq LK$.
This induces a natural ordering on $D_{l,k}$, illustrated in Fig. \ref{fig:f3},
that corresponds to the structure of refinement relations. We will assume for simplicity of presentation that $D_{0}=D_{1,0}$ is a connectivity data structure in which all vertices are connected.
For each edge $e$, $1 \le \ell \le L$, and $1 \le k \le K$, we define
an independent Bernoulli random variable $A'_{l, k, e}$ with $\prob[A'_{l, k, e}=1]=2^{-l}$.
The algorithm is as follows:
\begin{description}
\item[Algorithm 3 (A One-Pass Sparsifier)]
\item[Input:] Edges of $G$ streamed in adversarial order: $(e_1,\ldots, e_m)$.
\item[Output:] A sparsification $G'$ of $G$.
\item[Initialization:]  Set $E':=\emptyset$.
\item[1.] Set $t=1$.
\item[2.] For all $J=1,\ldots, LK$ ($J=(l,k)$)
\item[3.] Add $e_t=(u_t, v_t)$ to $D_{J}$ if $A'_{l,k,e}=1$ and $u_t$ and $v_t$ are connected in $D_{J-1}$.
\item[4.] Define $L'(e_t)$ as the minimum $l$ such that $u_t$ and $v_t$ are not connected in $D_{l, K}$. Set $z'(e_t):=\min\left\{1, \frac{4\rho}{\e^2 2^{L'(e_t)}}\right\}$. Output $e_t$ with probability $z'(e_t)$, giving it weight $1/z'(e_t)$.
\item[5.] Set $t:=t+1$. Go to step 2 if $t\leq m$.
\end{description}

Informally, Algorithm 3 underestimates strength of some edges until the data structures $D_{l, k}$ become properly connected but proceeds similarly to Algorithms 1 and 2 after that. Our main goal in the rest of the section is to show that this underestimation of strengths does not lead to a large increase in the size of the sample.

Note that not all $LK=\Theta(\log^2 n)$ coin tosses $A'_{l, k, e}$ per edge are necessary for an implementation of Algorithm 3 (in particular, we will show that Algorithm 3 can be implemented with $O(\log\log n)=o(LK)$ work per edge). However, the random variables $A'_{l,k,e}$ are useful for analysis purposes. We now show that Algorithm 3 outputs a sparsification $G'$ of $G$ with $O(n\log^3 n/\e^2)$ edges whp.

\begin{lemma} \label{lm:sparsification}
For any $\eps > 0$, w.h.p. the graph $G'$ is an $\eps$-sparsification of $G$.
\end{lemma}
\begin{proof}
We can couple behaviors of Algorithms 1 and 3 using
the coin tosses $A'_{l, k, e}$ to show that $L(e)\geq L'(e)$ for every edge $e$, i.e. $z'(e)\geq z(e)$. Hence $G'$ is a sparsification of $G$ by Corollary \ref{cor:oversampling}.
\end{proof}

It remains to upper bound the size of the sample.  The following lemma is crucial to our analysis; its proof is deferred to the Appendix \ref{app:xd}
due to space limitations.

\begin{lemma} \label{lm:xd}
Let $G(V, E)$ be an undirected graph. Consider the execution of Algorithm 3, and for
$1 \leq J \le LK$ where $J=(l,k)$, let $X^{J}$ denote
the set of edges $e=(u, v)$ such that $u$ and $v$ are connected in $D_{J-1}$ when $e$ arrives. Then $|E\setminus X^{J}|=O(K 2^l n)$ with high probability.
\end{lemma}

\begin{lemma} \label{lm:bound}
The number of edges in $G'$ is $O(n \log^3 n/\e^2)$ with high probability.
\end{lemma}
\begin{proof}
Recall that Algorithm 3 samples an edge $e_t=(u_t, v_t)$ with probability $z'(e_t)=\min\left\{1, \frac{4\rho}{\e^2 2^{L'(e_t)}}\right\}$, where
$L'(e_t)$ is the minimum $l$ such that $u_t$ and $v_t$ are not connected in $D_{l, K}$.
As before, for $J=(l,k)$, we denote by $X^J$ the set of edges $e=(u, v)$ such that $u$ and $v$ are connected in $D_{J-1}$ when $e$ arrives. 
Note that w.h.p. $X^{(L, 1)}=\emptyset$ w.h.p. by our choice of $L=\log(2n)$. For each $1\leq l\leq L$, let $Y_l=X^{(l, 1)}\setminus X^{(l+1, 1)}$. We have by Lemma~\ref{lm:xd} that $\sum_{1\leq j\leq l}|Y_j|=O(K2^ln)$ w.h.p. Also note that edges in $Y_l$ are sampled with probability at most $\frac{4\rho}{\e^2 2^{l-1}}$.
Hence, we get that the expected number of edges in the sample is at most
\begin{equation*}
\sum_{l=1}^{L} |Y_l|\cdot \frac{4\rho}{\e^2 2^{l-1}}=O\left(\sum_{l=1}^{L} K 2^l n\cdot \frac{4\rho}{\e^2 2^{l-1}}\right)=O(n\log^3 n/\e^2).
\end{equation*}

The high probability bound now follows by standard concentration inequalities.
\end{proof}

\smallskip
Finally, we have the following theorem.

\begin{theorem}
For any $\eps > 0$ and $d > 0$,
there exists a one-pass algorithm that given the edges of an undirected graph $G$ streamed in adversarial order, produces an $\eps$-sparsifier  $G'$ with $O(n\log^3 n/\e^2)$ edges with probability at least $1-n^{-d}$. The algorithm takes $O(\log\log n)$ amortized time per edge and uses $O(\log^2 n)$ space per node.
\end{theorem}
\begin{proof}
Lemma \ref{lm:sparsification} and Lemma \ref{lm:bound} together establish that
$G'$ is an $\eps$-sparsifier  $G'$ with $O(n\log^3 n/\e^2)$ edges.
It remains to prove the stated runtime bounds.

Note that when an edge $e_t=(u_t, v_t)$ is
processed in step 3 of Algorithm 3, it is  not necessary to add $e_t$ to any data structure $D_{J}$ in which $u_t$ and $v_t$ are already connected. Also, since $D_{J}$  is a refinement of $D_{J'}$ whenever $J'\leq J$, for every edge $e_t$ there exists $J^*$ such that $u_t$ and $v_t$ are connected in $D_{J}$  for any $J\leq J^*$ and not connected for any $J\geq J^*$. The value of $J^*$ can be found in $O(\log \log n)$ time by binary search. Now we need to keep adding $e_t$ to $D_{J}$,  for each  $J\geq J^*$ such that $U_{l, k, e_t}=1$. However, we have that $\expect\left[\sum_{J\geq J^*} U'_{l, k, e_t}\right]=O(1)$. Amortizing over all edges, we get $O(1)$ per edge using standard concentration inequalities.
\end{proof}

\begin{corollary} \label{cor:logn-sample}
For any $\eps > 0$ and $d > 0$,
there exists a one-pass algorithm that given the edges of an undirected graph $G$ streamed in adversarial order, produces an $\eps$-sparsifier  $G'$ with $O(n\log n/\e^2)$ edges with probability at least $1-n^{-d}$. The algorithm takes amortized
$O( \log\log n+(n/m)\log^5 n)$ time per edge and uses $O(\log^3 n)$ space per node.
\end{corollary}
\begin{proof}
One can obtain a sparsification of $G'$ with $O(n \log^3 n/\e^2)$ edges by running Algorithm 3 on the input graph $G$, and then run the Bencz\'{u}r-Karger algorithm
on $G'$ without violating the restrictions of the semi-streaming model. Note that even though $G'$ is a weighted graph, this will have overhead $O(\log^2 n)$ per edge of $G'$ since the weights are polynomial. Since $G'$ has $O(n\log^3 n)$ edges, the amortized work per edge of $G$ is $O(\log n\log\log n+(n/m)\log^5 n)$. The Bencz\'{u}r-Karger algorithm can be implemented using space proportional to the size of the graph, which yields $O(\log^3 n)$ space per node.
\end{proof}

\begin{remark}
The algorithm avove improves upon the runtime of the Bencz\'{u}r-Karger sparsification scheme when $m=\omega(n\log^3 n)$.
\end{remark}

\medskip
\noindent
{\bf Sparse $k$-connectivity Certificates:}
Our analysis of the performance of refinement sampling is along broadly similar lines to the analysis of the strength estimation routine
in \cite{benczurkarger96}. To make this analogy more precise, we note that refinement sampling as used in Algorithm 3 in fact produces a \emph{sparse connectivity certificate} of $G$, similarly to the algorithm of Nagamochi-Ibaraki\cite{nagamochi-ibaraki}, although with slightly weaker guarantees
on size.

A {\em $k$-connectivity certificate}, or simply a {\em $k$-certificate},
for an $n$-vertex graph $G$ is a subgraph $H$ of $G$ such that
contains all edges crossing cuts of size $k$ or less in $G$. Such a certificate
always exists with $O(kn)$ edges, and moreover, there are graphs where
$\Omega(kn)$ edges are necessary.
The algorithm of \cite{nagamochi-ibaraki} depends on random access to edges of $G$ to
produce a $k$-certificate with $O(kn)$ edges in $O(m)$ time.
We now show that refinement sampling gives a one-pass algorithm
to produce a $k$-certificate with $O(kn \log^2 n)$ edges in 
time $O(m\log\log n+n\log n)$. The result is summarized in the following corollary:

\begin{corollary}\label{cor:sparse-cert}
Whp for each $l\geq 1$ the set $X(D_{l, K})$ is a $2^l$-certificate of $G$
with $O(\log^2n)2^l n$ edges. 
\end{corollary}
\begin{proof}
Whp $X(D_{l, K})$ contains all $2^l$-weak edges, in particular those that cross cuts of size at most $2^l$. The bound on the size follows by Lemma~\ref{lm:xd}.
\end{proof}

\begin{remark}\label{rmk:weights}
Algorithms 1-3 can be easily extended to graphs with polynomially bounded integer weights on edges. 
If we denote by $W$ the largest edge weight, then it is sufficient to set the number of levels $L$ to $\log (2nW)$ instead of $\log (2n)$ and the number of passes to $\log_{4/3} nW$ instead of $\log_{4/3} n$.
A weighted edge is then viewed as several parallel edges, and sampling can be performed efficiently for such edges by sampling directly from the corresponding binomial distribution.
\end{remark}

\section{A Linear-time Algorithm for $O(n \log n/\e^2)$-size Sparsifiers} \label{sec:two-pass}

We now present an algorithm for computing an $\eps$-sparsification with $O(n \log n/\eps^2)$
edges in $O(m\log \frac1{\delta}+ n^{1+\delta})$ expected time for any $\delta>0$. Thus, the algorithm runs in linear-time whenever $m=\Omega(n^{1 + \Omega(1)})$.
We note that no (randomized) algorithm can output an $\eps$-sparsification in sub-linear time even if
there is no restriction on the size of the sparsifier. This is easily seen by considering the
family of graphs formed by disjoint union of two $n$-vertex graphs $G_1$ and $G_2$
with $m$ edges each, and a single edge $e$ connecting the two graphs.
The cut that separates $G_1$ from $G_2$ has a single edge $e$, and hence
any $\eps$-sparsifier must include $e$. On the other hand, it is easy to see that $\Omega(m)$
probes are needed in expectation to discover the edge $e$.

Our algorithm can in fact be viewed as a {\em two-pass} streaming algorithm, and we present is
as such below.
As before, let $G=(V, E)$ be an undirected unweighted graph. We will use Algorithm 3 as a building block of our construction. We now describe each of the passes.

\begin{description}
\item[First pass:] Sample every edge of $G$ uniformly at random with probability $p=4/\log n$. Denote the resulting graph by $G'=(V, E')$. Give the stream of sampled edges to Algorithm 3 as the input stream, and save the state of the connectivity data structures $D_{l, K}$ for all $1\leq l\leq L$ at the end of execution. For $1 \le l \le L$, let $D^*_{l}$ denote these connectivity data structures (we will also refer to $D^*_l$ as partitions in what follows).

\end{description}

Note that the first pass takes $O(m)$ expected time since Algorithm 3 has an overhead $O(\log\log n)$ time per edge and the expected size of $|E'|$ is $|E|/\log n$.

Recall that the partitions $D^*_l$ are used in Algorithm 3 to estimate strength of edges $e\in E'$. We now show that these partitions can also be used to estimate strength of edges in $E$. The following lemma establishes a relationship between the edge strengths in $G'$ and $G$.
For every edge $e \in E$, let $s'_e$ denote the strength of edge $e$ in the graph $G'_e(V, E'\cup \{e\})$.

\begin{lemma}
Whp $s'_e \le s_e \le 2 s_e'\log n + \rho\log n$ for all $e \in E$,
where $\rho=16(d+2)\ln n$ is the oversampling parameter in Karger sampling.
\end{lemma}
\begin{proof}
The first inequality is trivially true since $G'_e$ is a subgraph of $G$. For the second one, let us first
consider any edge $e\in E$ with $s_e > \rho \log n$. Let $C$ be the $s_e$-strong
component in $G$ that contains the edge $e$. By Karger's theorem, whp the capacity of any cut
defined by a partition of vertices in $C$ decreases by a factor of at most $2\log n$ after sampling edges of $G$ with probability $p=4/\log n=\rho/((1/2)^2\rho\log n)$, i.e. in going
from $G$ to $G'$. So any cut in $C$, restricted to edges in $E'$ has size at least $s_e/(2\log n)$, implying that $s'_e\geq s_e/(2\log n)$. Finally, for any edge $e$ with $s_e\leq \rho\log n$, $s'_e$ is at least $1$, and the inequality thus follows.
\end{proof}

We now discuss the second pass over the data.
Recall that in order to estimate the strength $s'_e$ of an edge $e\in E'$, Algorithm 3 finds the minimum $L(e)$ such that the endpoints of $e$ are not connected in $D^*_{l}$ by doing a binary search over the range $[1..L]$. For an edge $e\in G$ we estimate its strength in $G'_e$ by doing binary
search as before, but stopping the binary search as soon as the size of the interval is smaller than $\delta L$, thus taking $O(\log\frac{1}{\delta})$ time per edge and obtaining an estimate that is away from the true value by a factor of at most $n^{\delta}$. Let $s''_e$ denote this estimate, that is, $s'_e n^{-\delta} \le s''_e \le s'_e n^\delta$. Now sampling every edge with probability $p_e=\min\left\{\frac{\rho n^{\delta}}{\e^2 s''_e}, 1\right\}$  and giving each sampled edge weight $1/p_e$
yields an $\e$-sparsification $G''=(V, E'')$ of $G$ whp.
Moreover, we have that w.h.p. $|E''|=\tilde{O}(n^{1+\delta})$.
Finally, we provide the graph $G''$ as input to Algorithm 3 followed by applying Bencz\'{u}r-Karger sampling as outlined in Corollary \ref{cor:logn-sample},
 obtaining a sparsifier of size $O(n\log n/\e^2)$.
We now summarize the second pass.

\begin{description}
\item[Second pass:] For each edge $e$ of the input graph $G$:
\begin{itemize}
\item Perform $O(\log\frac1{\delta})$ steps of binary search to calculate $s_e''$.
\item
Sample edge $e$ with probability $p_e=\min\{\frac{\rho n^{\delta}}{\e^2 s''_e}, 1\}$.
\item
If $e$ is sampled, assign it a weight of $1/p_e$, and pass it as an input
to a fresh invocation of Algorithm 3, followed by Bencz\'{u}r-Karger sampling as
outlined in Corollary \ref{cor:logn-sample}, giving the final sparsification.
\end{itemize}
\end{description}
Note that the total time taken in the second pass is $O(m\log\frac1{\delta})+\tilde O(n^{1+\delta})$.
We have proved the following

\begin{theorem}
For any $\eps > 0$ and $\delta > 0$, there exists a two-pass algorithm that produces
an $\e$-sparsifier in time $O(m\log\frac1{\delta})+\tilde O(n^{1+\delta})$. Thus the algorithm
runs in linear-time when $m=\Omega(n^{1+\delta})$ and $\delta$ is constant.
\end{theorem}
\newpage
\pagenumbering{Roman} 

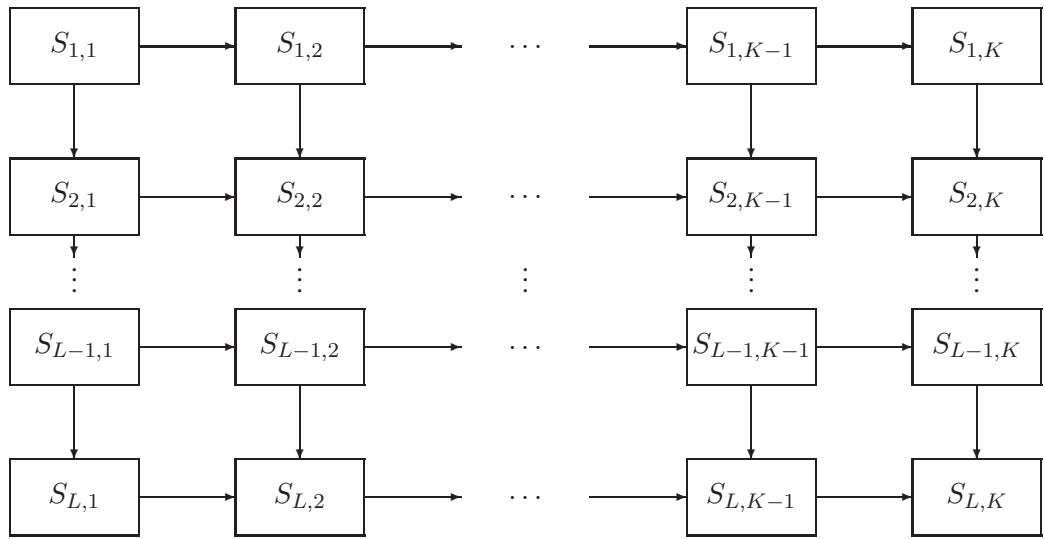
\begin{figure}
\begin{picture}(15,9)(0,0)
\put(1, 8){\framebox(1.7, 1){$S_{1,1}$}}
\put(4, 8){\framebox(1.7, 1){$S_{1,2}$}}
\put(7, 8){\makebox(1.7, 1){$\ldots$}}
\put(10, 8){\framebox(1.7, 1){$S_{1,K-1}$}}
\put(13, 8){\framebox(1.7, 1){$S_{1,K}$}}

\put(2.7,8.5){\vector(1, 0){1.3}}
\put(5.7,8.5){\vector(1, 0){1.3}}
\put(8.7,8.5){\vector(1, 0){1.3}}
\put(11.7,8.5){\vector(1, 0){1.3}}

\put(1.85,8){\vector(0, -1){1}}
\put(4.85,8){\vector(0, -1){1}}
\put(10.85,8){\vector(0, -1){1}}
\put(13.85,8){\vector(0, -1){1}}

\put(1, 6){\framebox(1.7, 1){$S_{2,1}$}}
\put(4, 6){\framebox(1.7, 1){$S_{2,2}$}}
\put(7, 6){\makebox(1.7, 1){$\ldots$}}
\put(10, 6){\framebox(1.7, 1){$S_{2,K-1}$}}
\put(13, 6){\framebox(1.7, 1){$S_{2,K}$}}

\put(2.7,6.5){\vector(1, 0){1.3}}
\put(5.7,6.5){\vector(1, 0){1.3}}
\put(8.7,6.5){\vector(1, 0){1.3}}
\put(11.7,6.5){\vector(1, 0){1.3}}

\put(1.85,6){\vector(0, -1){0.3}}
\put(4.85,6){\vector(0, -1){0.3}}
\put(10.85,6){\vector(0, -1){0.3}}
\put(13.85,6){\vector(0, -1){0.3}}

\put(1, 5){\makebox(1.7, 1){$\vdots$}}
\put(4, 5){\makebox(1.7, 1){$\vdots$}}
\put(7, 5){\makebox(1.7, 1){$\vdots$}}
\put(10, 5){\makebox(1.7, 1){$\vdots$}}
\put(13, 5){\makebox(1.7, 1){$\vdots$}}

\put(1, 4){\framebox(1.7, 1){$S_{L-1,1}$}}
\put(4, 4){\framebox(1.7, 1){$S_{L-1,2}$}}
\put(7, 4){\makebox(1.7, 1){$\ldots$}}
\put(10, 4){\framebox(1.7, 1){$S_{L-1,K-1}$}}
\put(13, 4){\framebox(1.7, 1){$S_{L-1,K}$}}

\put(2.7,4.5){\vector(1, 0){1.3}}
\put(5.7,4.5){\vector(1, 0){1.3}}
\put(8.7,4.5){\vector(1, 0){1.3}}
\put(11.7,4.5){\vector(1, 0){1.3}}

\put(1.85,4){\vector(0, -1){1}}
\put(4.85,4){\vector(0, -1){1}}
\put(10.85,4){\vector(0, -1){1}}
\put(13.85,4){\vector(0, -1){1}}

\put(1, 2){\framebox(1.7, 1){$S_{L,1}$}}
\put(4, 2){\framebox(1.7, 1){$S_{L,2}$}}
\put(7, 2){\makebox(1.7, 1){$\ldots$}}
\put(10, 2){\framebox(1.7, 1){$S_{L,K-1}$}}
\put(13, 2){\framebox(1.7, 1){$S_{L,K}$}}

\put(2.7,2.5){\vector(1, 0){1.3}}
\put(5.7,2.5){\vector(1, 0){1.3}}
\put(8.7,2.5){\vector(1, 0){1.3}}
\put(11.7,2.5){\vector(1, 0){1.3}}

\end{picture}
\caption{Scheme of refinement relations for Algorithm 2.}
\label{fig:f2}
\end{figure}

\begin{figure}
\begin{picture}(20,10)(0,0)
\put(1, 8){\framebox(1.7, 1){$S_{1,1}$}}
\put(4, 8){\framebox(1.7, 1){$S_{1,2}$}}
\put(7, 8){\makebox(1.7, 1){$\ldots$}}
\put(10, 8){\framebox(1.7, 1){$S_{1,K-1}$}}
\put(13, 8){\framebox(1.7, 1){$S_{1,K}$}}

\put(13.85, 8){\line(0, -1){0.5}}
\put(1.85, 7.5){\line(1, 0){12}}
\put(1.85, 7.5){\vector(0, -1){0.5}}

\put(2.7,8.5){\vector(1, 0){1.3}}
\put(5.7,8.5){\vector(1, 0){1.3}}
\put(8.7,8.5){\vector(1, 0){1.3}}
\put(11.7,8.5){\vector(1, 0){1.3}}

\put(1, 6){\framebox(1.7, 1){$S_{2,1}$}}
\put(4, 6){\framebox(1.7, 1){$S_{2,2}$}}
\put(7, 6){\makebox(1.7, 1){$\ldots$}}
\put(10, 6){\framebox(1.7, 1){$S_{2,K-1}$}}
\put(13, 6){\framebox(1.7, 1){$S_{2,K}$}}

\put(13.85, 6){\line(0, -1){0.15}}
\put(1.85, 5.85){\line(1, 0){12}}
\put(1.85, 5.85){\vector(0, -1){0.15}}

\put(13.85, 5.35){\line(0, -1){0.15}}
\put(1.85, 5.15){\line(1, 0){12}}
\put(1.85, 5.15){\vector(0, -1){0.15}}

\put(2.7,6.5){\vector(1, 0){1.3}}
\put(5.7,6.5){\vector(1, 0){1.3}}
\put(8.7,6.5){\vector(1, 0){1.3}}
\put(11.7,6.5){\vector(1, 0){1.3}}

\put(1, 5){\makebox(1.7, 1){$\vdots$}}
\put(4, 5){\makebox(1.7, 1){$\vdots$}}
\put(7, 5){\makebox(1.7, 1){$\vdots$}}
\put(10, 5){\makebox(1.7, 1){$\vdots$}}
\put(13, 5){\makebox(1.7, 1){$\vdots$}}

\put(1, 4){\framebox(1.7, 1){$S_{L-1,1}$}}
\put(4, 4){\framebox(1.7, 1){$S_{L-1,2}$}}
\put(7, 4){\makebox(1.7, 1){$\ldots$}}
\put(10, 4){\framebox(1.7, 1){$S_{L-1,K-1}$}}
\put(13, 4){\framebox(1.7, 1){$S_{L-1,K}$}}

\put(13.85, 4){\line(0, -1){0.5}}
\put(1.85, 3.5){\line(1, 0){12}}
\put(1.85, 3.5){\vector(0, -1){0.5}}

\put(2.7,4.5){\vector(1, 0){1.3}}
\put(5.7,4.5){\vector(1, 0){1.3}}
\put(8.7,4.5){\vector(1, 0){1.3}}
\put(11.7,4.5){\vector(1, 0){1.3}}

\put(1, 2){\framebox(1.7, 1){$S_{L,1}$}}
\put(4, 2){\framebox(1.7, 1){$S_{L,2}$}}
\put(7, 2){\makebox(1.7, 1){$\ldots$}}
\put(10, 2){\framebox(1.7, 1){$S_{L,K-1}$}}
\put(13, 2){\framebox(1.7, 1){$S_{L,K}$}}

\put(2.7,2.5){\vector(1, 0){1.3}}
\put(5.7,2.5){\vector(1, 0){1.3}}
\put(8.7,2.5){\vector(1, 0){1.3}}
\put(11.7,2.5){\vector(1, 0){1.3}}

\end{picture}
\caption{Scheme of refinement for Algorithm 3.}
\label{fig:f3}
\end{figure}
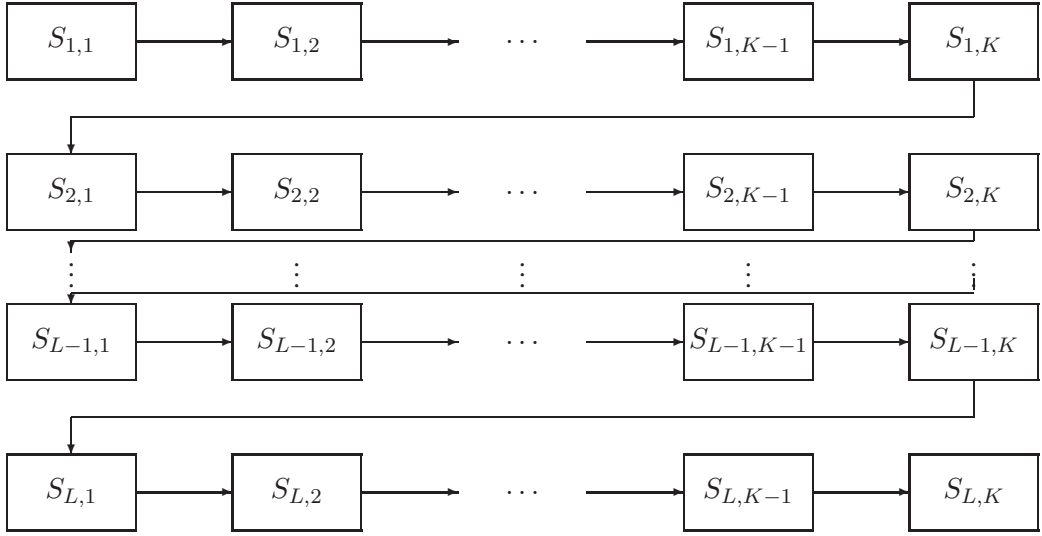

\newpage
\pdfbookmark[1]{References}{MyRefs}

\begin{thebibliography}{FKM{\etalchar{+}}05}

\bibitem[AG09]{anh-guha}
K.~Ahn and S.~Guha.
\newblock On graph problems in a semi-streaming model.
\newblock {\em Automata, languages and programming: Algorithms and complexity},
  pages 207 -- 216, 2009.

\bibitem[AS08]{b:alonspencer}
N.~Alon and J.~Spencer.
\newblock {\em The probabilistic method}.
\newblock Wiley, 2008.

\bibitem[BK96]{benczurkarger96}
Andr{\'a}s~A. Bencz{\'u}r and David~R. Karger.
\newblock Approximating {\it s-t} minimum cuts in $\tilde{O}(n^2)$ time.
\newblock {\em Proceedings of the 28th annual ACM symposium on Theory of
  computing}, pages 47--55, 1996.

\bibitem[CLRS01]{b:clr}
T.~Cormen, C.~Leiserson, R.~Rivest, and C.~Stein.
\newblock {\em Introduction to Algorithms}.
\newblock MIT Press and McGraw-Hill, 2001.

\bibitem[FKM{\etalchar{+}}05]{fkmsz05}
J.~Feigenbaum, S.~Kannan, A.~McGregor, S.~Suri, and J.~Zhang.
\newblock On graph problems in a semi-streaming model.
\newblock {\em Theor. Comput. Sci.}, 348:207--216, 2005.

\bibitem[KL02]{kl02}
D.~Karger and M.~Levine.
\newblock Random sampling in residual graphs.
\newblock {\em STOC}, 2002.

\bibitem[KRV06]{krv06}
R.~Khandekar, S.~Rao, and V.~Vazirani.
\newblock Graph partitioning using single commodity flows.
\newblock {\em STOC}, pages 385 -- 390, 2006.

\bibitem[Mut06]{b:streaming}
S.~Muthukrishnan.
\newblock {\em Data streams: algorithms and applications}.
\newblock Now publishers, 2006.

\bibitem[NI92]{nagamochi-ibaraki}
H.~Nagamochi and T.~Ibaraki.
\newblock Computing edge-connectivity in multigraphs and capacitated graphs.
\newblock {\em SIAM Journal on Discrete Mathematics}, 5(1):54--66, 1992.

\bibitem[SS08]{ss:sample2008}
D.A. Spielman and N.~Srivastava.
\newblock Graph sparsification by effective resistances.
\newblock {\em STOC}, pages 563--568, 2008.

\end{thebibliography}
\newcommand{\etalchar}[1]{$^{#1}$}

\newpage
\appendix

\section{Proof of Lemma \ref{lm:xd}}
\label{app:xd}
We denote the edges of $G$ in their order in the stream by $E=(e_1,\ldots, e_m)$. In what follows we shall treat edge sets as ordered sets, and for any
$E_1\subseteq E$ write $E\setminus E_1$ to denote the result of removing edges of $E_1$ from $E$ while preserving the order of the remaining edges. For a stream of edges $E$ we shall write $E_t$ to denote the set of the first $t$ edges in the stream.

For a $\kappa$-connected component $C$ of a graph $G$ we will write $|C|$ to denote the number of vertices in $C$. Also, we will denote the result of sampling the edges of $C$ uniformly at random with probability $p$ by $C'$. The following simple lemma will be useful in our analysis:

\begin{lemma} \label{lm:components-weak}
Let $C$ be a $\kappa$-connected component of $G$ for some positive integer $\kappa$. Denote the graph obtained by sampling edges of $C$ with probability $p\geq \lambda/\kappa$ by $C'$. Then the number of connected components in $C'$ is at most $\gamma |C|$ with probability at least $1-e^{-\eta|C|}$, where $\gamma=(7/8+e^{-\lambda/2}/8)$ and $\eta=1-e^{-\lambda/2}$.
\end{lemma}
\begin{proof}
Choose $A, B\subset V(C)$ so that $A\cup B=V(C), A\cap B=\emptyset$, $|A|\geq |V(C)|/2$ and for every $v\in A$ at least half of its edges that go to vertices in $C$ go to $B$.
Note that such a partition always exists: starting from any arbitrary partition of vertices of $C$,
we can repeatedly move a vertex from one side to the other if it increases the number of edges
going across the partition, and upon termination, the larger side corresponds to the set $A$.
Denote by $Y$ the number of vertices of $A$ that belong to components of size at least $2$.
Note that $Y$ can be expressed as sum of $|A|$ independent $0/1$ Bernoulli random
variables. Let $\mu:=\expect[Y]$; we
have that $\mu \geq |A|(1-(1-\lambda/\kappa)^{\kappa/2})\geq |A|(1-e^{-\lambda/2})$.  We get by the Chernoff bound that $\prob[Y\leq |A|(1-e^{-\lambda/2})/2]\leq e^{-2\mu}\leq e^{-|C|(1-e^{-\lambda/2})}=e^{-\eta |C|}$.  Hence, at least a $(1-e^{-\lambda/2})/4$ fraction of the vertices of $C$ are in components of size at least $2$. Hence, the number of connected components is at most a $1-(1-e^{-\lambda/2})/8=7/8+e^{-\lambda/2}/8=\gamma$ fraction of the number of vertices of $C$.
\end{proof}

\begin{proofof}{Lemma \ref{lm:xd}}
The proof is by induction on $J$. We prove that w.h.p. for every $J=(l,k)$ one has $|E\setminus X^J|\leq \sum_{1\leq J'=(l',k')\leq J-1} c_1 2^{l'}n$ for a constant $c_1>0$.
\begin{description}
\item[Base: $J=1$] Since everything is connected in $D_{0}$ by definition, the claim holds.
\item[Inductive step: $J\to J+1$] The outline of the proof
is as follows. For every $J=(l,k)$ we consider the edges of the stream that the algorithm tries to add to $D_J$, identify a sequence of $2^l$-strongly connected components $C_0,C_1\ldots$ in the partially received graph, and use lemma \ref{lm:components-weak} to show that the number of connected components decreases fast because only a small fraction of vertices in the sampled $2^l$-strongly connected components are isolated. We thus show that, informally, it will take $O(2^ln)$ edges to make the connectivity data structure $D_{J}$ in Algorithm 3 connected. The connected components $C_s$ are defined by induction on $s$. The vertices of $C_s$ are elements of a partition $P_s$ of the vertex set $V$ of the graph $G$. We shall use an auxiliary sequence of graphs which we denote by $H_s^t$.

 Let $P_0$ be the partition consisting of isolated vertices of $V$. We treat the base case $s=0$ separately to simplify exposition. We use the definition of $\gamma$ and $\eta$ from lemma \ref{lm:components-weak} with $\lambda=1$ since we are considering $2^l$-connected components when $J=(k,l)$.
\begin{description}
\item[Base case: $s=0$.] Set $H_0^t=(P_0, \{e_1,\ldots, e_t\})$, i.e. $H_0^t$ is the partially received graph up to time $t$. Let $t_0^*$ be the the first value of $t$ such that $s_{H_0^{t}}(e_t)\geq 2^{l}$. This means that $e_{t_0^*}$ belongs to a $2^{l}$-strongly connected component in $H_0^{t_0^*}$. Note that this component does not contain any $(2^{l}+1)$-strongly connected components. Denote this component by $C_0$ (note that the number of edges in $C_0$ is at most $2^{l}|C_0|$ by lemma \ref{lm:k-weak}). Denote the random variables that correspond to sampling edges of $C_0$ by $R_0$. Let $X_0$ be an indicator variable that equals $1$ if the number of connected components in $C'_0$ is at most $\gamma|C_0|$ and $0$ otherwise. By lemma \ref{lm:components-weak} we have that $\prob[X_0=1]\geq 1-e^{-\eta|C_0|}$.

For a partition $P$  denote $\text{diag}(P)=\{(u, u): u\in P\}$. Define $P_1$ by merging partitions of $P_0$ that belong to connected components in $C'_0$ if $X_0=1$ and as equal to $P_0$ otherwise. Let $E^1=E\setminus (E(C_0)\cup \text{diag}(P_1))$, i.e. we remove edges of $C_0$ and also edges that connect vertices that belong to the same partition in $P_1$.
Note that we can safely remove these edges since their endpoints are connected in $D_J$ when they arrive. Define $H_1^t=(P_1, E^1_t)$, i.e. $H_1^t$ is the partially received graph on the modified stream of edges.

\item[Inductive step: $s\to s+1$.]
As in the base case, let $t_{s}^*$ be the the first value of $t$ such that $s_{H_{s}^{t}}(e_t)\geq 2^{l}$. This means that $e_{t_{s}^*}$ belongs to a $2^{l}$-connected component in $H_{s}^{t_{s}^*}$. Denote this component by $C_{s}$(note that the number of edges in $C_{s}$ is at most $2^{l}|C_{s}|$ by lemma \ref{lm:k-weak}). Denote the random variables that correspond to sampling edges of $C_s$ by $R_s$. Let $X_{s}$ be an indicator variable that equals $1$ if the number of connected components in $C'_{s}$ is at most $\gamma |C_{s}|$ and $0$ otherwise. By lemma \ref{lm:components-weak} we have that $\prob[X_{s}=1]\geq 1-e^{-\eta|C_{s}|}$. Define $P_{s+1}$ by merging together vertices that belong to connected components in $C'_{s}$. Let $E^{s+1}=E^{s}\setminus (E(C_{s})\cup \text{diag}(P_s))$. Denote
$H_s^t=(P_s, E^{s}_t)$.
\end{description}

It is important to note that at each step $s$ we only flip coins $R_s$ that correspond to edges in $E(C_{s})$, and delete only those edges from $E^{s}$.
While there may be edges going across partitions $P_s$ for which we do not perform a coin flip,
there number is bounded by $O(2^ln)$ since these edges do not contain a $2^l$-connected component.

Note that for any $s>0$ the number of connected components in $P_s$ is at most
\begin{equation*}
n-\sum_{j=1}^s (1-\gamma) |C_j| X_j.
\end{equation*}

We now show that it is very unlikely that $\sum_{j=1}^s|C_j| X_j$ is more than a constant factor smaller than $\sum_{j=1}^s|C_j|$, thus showing that the number of connected components cannot be more than $1$ when $\sum_{j=1}^s|C_j|\geq \frac{cn}{1-\gamma}$ for an appropriate constant $c>0$.

For any constant $d>0$ define $I^+=\{i\geq 0: |C_i|>((d+2)/\eta)\log n\}$ and $I^-=\{i\geq 0: |C_i|\leq ((d+2)/\eta)\log n\}$.
Also define $Z_i^+=\sum_{0\leq j\leq i, j\in I^+} X_j |C_j|, Z_i^-=\sum_{0\leq j\leq i, j\in I^-} X_j|C_j|-|C_j|(1-e^{-\eta |C_j|})$.

First note that one has $\prob[X_j=1]\geq 1-n^{-d-2}$ for any $j\in I^+$ by lemma \ref{lm:components-weak}. Hence, it follows by taking the union bound that $i\leq n^2$ one has $\prob[Z_i^+=\sum_{j\in I^+, j\leq i} |C_j|]\geq 1-n^{-d}$.

We now consider $Z_i^-$. Note that $Z_i^-$'s define a martingale sequence with respect to $R_{i-1},\ldots, R_0$: $\expect[Z_i^-|R_{i-1},\ldots, R_0]=Z_{i-1}^-$.
Also, $|Z_{i}^--Z_{i-1}^-|\leq ((d+2)/\eta)\log n$ for all $i$. Hence, by Azuma's inequality (see, e.g. \cite{b:alonspencer}) one has
\begin{equation*}
\prob[Z_{i}^-<t]<\exp\left(-\frac{t^2}{2 i (((d+2)/\eta)\log n)^2}\right).
\end{equation*}
Now consider the smallest value $\tau$ such that $\sum_{j\leq \tau} |C_j|=\sum_{j\leq \tau,j\in I^+} |C_j|+\sum_{j\leq i,j\in I^-} |C_j|=S^++S^-\geq \frac{4n}{(1-e^{-2\eta})(1-\gamma)}$.
Note that $\tau<n/(2(1-e^{-2\eta})(1-\gamma))$ since $|C_i|\geq 2$. If $S^+\geq \frac{2n}{(1-e^{-2\eta})(1-\gamma)}\geq 2n/(1-\gamma)$, then we have that $Z_\tau^+=S^+>2n/(1-\gamma)$ with probability at least $1-n^{-d}$. Thus,
\begin{equation*}
n-\sum_{j=1}^\tau (1-\gamma) |C_j| X_j\leq n-(1-\gamma) Z_\tau^+\leq 0.
\end{equation*}
Otherwise $S^-\geq \frac{2n}{(1-e^{-2\eta})(1-\gamma)}$ and by Azuma's inequality we have
\begin{equation*}
\prob[Z_{\tau}^-<-n]<\exp\left(-\frac{n^2}{2 \tau (((d+2)/\eta)\log n)^2}\right)\leq \exp\left(-\frac{n}{(((d+2)/\eta)\log n)^2}\right)<n^{-d}.
\end{equation*}
Since $|C_i|\geq 2$, we have $|C_i|(1-e^{-\eta |C_i|})\geq |C_i|(1-e^{-2\eta})$ and thus we get
\begin{equation*}
\begin{split}
n-\sum_{j=1}^\tau (1-\gamma) |C_j| X_j<n-(1-\gamma)  \left[\sum_{1\leq j\leq \tau, j\in I^-} |C_j|(1-e^{-\eta |C_j|})+Z_\tau^-\right]\\
<n-(1-\gamma)  \left[(1-e^{-2\eta})\sum_{1\leq j\leq \tau, j\in I^-} |C_j|+Z_\tau^-\right]\\
<n-(1-\gamma)  \left[\frac{2n}{1-\gamma}+n\right]<0\\
\end{split}
\end{equation*}

We have shown that there exists a constant $c'>0$ such that with probability at least $1-n^{-d}$ after $c' 2^l n$ edges are sampled by the algorithm at level $J$ all subsequent edges will have their endpoints connected in $D_J$. Note that we never flipped coins for those edges that did not contain a $2^{l}$-connected component.
Setting $c_1=c'+1$, we have that w.h.p. $|E\setminus X^J|\leq c_1 2^l n+|E\setminus X^{J-1}|$. By the inductive hypothesis we have that $|E\setminus X^{J-1}|\leq \sum_{1\leq J'=(l',k')\leq J-2} c_1 2^{l'}n$, which together with the previous estimate gives us the desired result.
\end{description}

It now follows that $|E\setminus X^J|\leq \sum_{1\leq J'=(l',k')\leq J-1} c_1 2^{l'}n=O(K2^ln)$ w.h.p., finishing the proof of the lemma.

\end{proofof}
\end{document}